\documentclass[twoside]{IEEEtran}
\usepackage{cite}
\usepackage{graphicx}
\usepackage{xcolor}
\usepackage{psfrag}
\usepackage{url}
\usepackage{amsmath}
\usepackage{array}
\usepackage{amssymb}
\usepackage{amsfonts}
\usepackage{epstopdf}
\newtheorem{proposition}{Proposition}
\newtheorem{lemma}{Lemma}

\newtheorem{corollary}{Corollary}
\newtheorem{remark}{Remark}

\newtheorem{theorem}{Theorem}

\usepackage{subfigure}

\title{Non-Adaptive Group Testing with Inhibitors}
\author{
\authorblockN{Abhinav Ganesan, Javad Ebrahimi, Sidharth Jaggi and Venkatesh Saligrama, \IEEEmembership{Senior Member, IEEE}}
\thanks{A. Ganesan, J. Ebrahimi and S. Jaggi are with the Chinese University of Hong Kong,
Hong Kong (e-mail: abhinav@inc.cuhk.edu.hk; javad.ebrahimi@gmail.com; jaggi@ie.cuhk.edu.hk).
V. Saligrama is with Boston University, Boston, USA (e-mail: srv@bu.edu).}
}

\begin{document}

\maketitle
\thispagestyle{empty}	

\begin{abstract}
Group testing with inhibitors (GTI) introduced by Farach at al. is studied in this paper. There are three types of items, $d$ defectives, $r$ inhibitors and $n-d-r$ normal items in a population of $n$ items. The presence of any inhibitor in a test can prevent the expression of a defective. For this model, we propose a probabilistic non-adaptive pooling design with a low complexity decoding algorithm. We show that the sample complexity of the number of tests required for guaranteed recovery with vanishing error probability using the proposed algorithm scales as $T=O(d \log n)$ and $T=O(\frac{r^2}{d}\log n)$ in the regimes $r=O(d)$ and $d=o(r)$ respectively. In the former regime, the number of tests meets the lower bound order while in the latter regime, the number of tests is shown to exceed the lower bound order by a $\log \frac{r}{d}$ multiplicative factor. When only upper bounds on the number of defectives $D$ and the number of inhibitors $R$ are given instead of their exact values, the sample complexity of the number of tests  using the proposed algorithm scales as $T=O(D \log n)$ and $T=O(R^2 \log n)$ in the regimes $R^2=O(D)$ and $D=o(R^2)$ respectively. In the former regime, the number of tests meets the lower bound order while in the latter regime, the number of tests exceeds the lower bound order by a $\log R$ multiplicative factor. The time complexity of the proposed decoding algorithms scale as $O(nT)$.
\end{abstract}		
\section{Introduction} \label{sec1}

Group testing, introduced by Dorfman \cite{Dor}, conventionally dealt with classifying unhealthy samples (items) called defectives and healthy ones in a huge population using a small number of tests. Classical group testing assumes binary outcome, i.e, the outcome of a test is positive if a defective item is present in a test and negative otherwise. Two categories of group testing problems are studied in the literature, namely probabilistic group testing (PGT) \cite{Dor} and combinatorial group testing (CGT) \cite{book-HD}. In the former, a probability distribution on the number of defectives is assumed while in the latter, the number of defectives is fixed or an upper bound on the number of defectives is known. This paper deals with CGT in the context of Group testing with inhibitors (GTI). The philosophy of group testing is that when the number of defective items is small, by pooling the items carefully, the items can be classified in a relatively small number of tests than when the items are tested individually.  

Two kinds of pooling designs have been vastly studied in the classical group testing literature - non-adaptive and $k$-stage adaptive pooling designs. In non-adaptive pooling designs, the pools are constructed all at once and tested parallely. This kind of pooling design is known to be economical as well as saves time in testing and they are of concern in biological applications \cite{BBTK}. A $k$-stage adaptive pooling design is comprised of pool construction and testing in $k$-stages, where the pools constructed for testing in the $k^{\text{th}}$ stage depend on the outcomes in the previous stages. A trivial two-stage adaptive pooling design was introduced by Knill \cite{Kni} in the PGT context. A trivial two-stage adaptive pooling design comprises two stages. In the first stage, a superset of the defectives is declared, and in the second stage, confirmatory tests are done. A trivial two-stage algorithm was proposed in the CGT context in \cite{DGV} and it was shown to perform asymptotically as good as the best adaptive algorithm.

GTI was introduced in \cite{FKKM} motivated by complications in blood testing \cite{PhS} and drug discovery \cite{LOGY} where blocker compounds (i.e., inhibitors) block the detection of potent compounds (i.e., defectives). GTI involves three types of items - defectives, normal items and inhibitors. A test is positive iff there is at least one defective and no inhibitors in the test. A randomized fully adaptive algorithm\footnote{A fully adaptive algorithm is where every pool constructed depends on its previous outcomes.} was proposed in \cite{FKKM} to identify upto $R$ inhibitors and upto $D$ defectives amidst $n$ items. Identifying both the defectives and the inhibitors was termed as Sample Classification Problem (SCP). Many of the later works on GTI focussed on identifying the defectives alone, termed as Defective Classification Problem (DCP). However, identifying both the defectives and the inhibitors is of interest for the following reason. For example, not all species of a given virus might be pathogenic\footnote{For instance, out of five known species of ebolavirus, only four of them are pathogenic to humans (see p. $5$ in \cite{Cli})}. The pathogenic proteins in the context of this paper are represented by the defective items and the non-pathogenic ones by the normal items. Initial stages of drug discovery attempt to find blocker or lead compounds amidst billions of chemical compounds \cite{XTSY,Cli}. These lead compounds which are referred to as the inhibitors are ultimately used to produce new drugs. Thus, the SCP intends to unify the process of finding both the pathogenic proteins and the lead compounds.

The focus of this paper is on $\epsilon$-error non-adaptive pooling design for SCP and DCP in the GTI model. Pooling designs and decoding algorithms are primarily proposed for SCP in this paper and the solution to DCP happens as a by-product.
\begin{table*}\caption{Summary of some known lower bounds and upper bounds on the number of tests using deterministic and probabilistic pooling designs. The terminology ``Deterministic'' refers to zero-error pooling designs and ``Probabilistic'' refers to $\epsilon$-error pooling designs. The second column indicates assumptions on knowledge of exact or upper bound on the number of inhibitors and defectives. The acronym UB stands for upper bound. This work primarily compares with \cite{CCF} where knowledge of the exact number of inhibitors and defectives was assumed. The lower bound for $\epsilon$-error SCP obtained in this work is different from the one for zero-error SCP in \cite{CCF} with non-adaptive pooling design.} \scriptsize
\begin{tabular}{|c|c|c|c|c|}
\hline
Reference, & No. of inhibitors ($r$)/ & Lower Bound & Upper Bound & Error/Pooling Design\\
DCP/ SCP&  Defectives ($d$) & & & \\
\hline
\cite{FKKM}, SCP & UB/ UB & $\Omega((R+D)\log n)$ &  & $\epsilon$-error, Adaptive\\
SCP & UB/ UB &  & $O((R+D)\log n)$ & Fully adaptive, Probabilistic\\
\hline
\cite{ADB}, DCP & UB/ UB &$\Omega\left(\frac{(R+D)^2}{\log (R+D)}\log n\right)$ && Zero-error, Non-adaptive\\
SCP&UB/ UB &$\Omega\left(\frac{R^2}{\log R}\log n + D \log n\right)$ && Zero-error, $k$-stage adaptive\\
DCP& UB/ Exact &  $\Omega\left(\frac{R^2}{d \log R}\log n+R\log n + d\log n\right)$, if $R \geq 2d$,& & Zero-error, $k$-stage adaptive\\&& $\Omega\left(R\log n +d\log n\right)$, if $R < 2d$. & & \\
DCP& UB/ Exact & & $O\left(\frac{R^2}{d}\log n+ d\log n\right)$, if $R \geq 2d$, & Trivial two-stage adaptive,\\&& &$O\left(R\log n +d\log n\right)$, if $R < 2d$.  &Deterministic. \\
\hline
\cite{DMTV}, DCP & UB/ UB & & $O\left((R+D)^2\log n\right)$ & Non-adaptive, Deterministic\\
\hline
\cite{DeV}, SCP & UB/ UB & & $O(R^2 \log n + D \log n )$ & $3$-Stage adaptive, Deterministic\\
\hline
\cite{DGV}, DCP & UB/ Exact & & $O\left(\frac{(R+d)^2}{d} \log n\right)$ & $4$-Stage adaptive, Deterministic\\
\hline
\cite{CCF}, DCP & Exact/ Exact &$\Omega\left(\frac{(r+d)^2}{\log (r+d)}\log n\right)$ & & Zero-error, Non-adaptive\\
 DCP & Exact/ Exact && $O\left((r+d)^2\log n\right)$& Non-adaptive, Deterministic\\
SCP& Exact/ Exact &$\max\left\{\Omega\left(\frac{(r+d)^2}{\log (r+d)}\log n\right), \Omega\left(\frac{r^3}{\log r}\log n\right)\right\}$&  & Zero-error, Non-adaptive\\
SCP& Exact/ Exact & &$O((r+d)^3 \log n)$ & Non-adaptive, Deterministic\\ 
\hline
\textbf{This work},& & $\Omega\left(d \log n\right),r=O(d)$&&$\epsilon$-error, Non-adaptive\\
SCP & Exact/ Exact & $\Omega\left(\frac{r^2}{d \log \frac{r}{d}}\log n\right), d=o(r)$ & &$\epsilon$-error, Non-adaptive\\
 & & & $O\left(d \log n\right), r=O(d)$ & Non-adaptive, Probabilistic\\
 & & & $O\left(\frac{r^2}{d} \log n\right), d=o(r)$ & Non-adaptive, Probabilistic\\
DCP &Exact/ Exact & $\Omega\left(d \log n\right),r=O(d)$&&$\epsilon$-error, Non-adaptive\\
 &  & $\Omega\left(\frac{r}{\log \frac{r}{d}}\log n\right), d=o(r)$ & &$\epsilon$-error, Non-adaptive\\
 & & & $O\left(d \log n\right), r=O(d)$ & Non-adaptive, Probabilistic\\
 & & & $O\left({r} \log n\right), d=o(r)$ & Non-adaptive, Probabilistic\\
SCP & UB/ UB & $\max\left\{\Omega\left((R+D) \log n\right),\Omega\left(\frac{R^2}{\log {R}}\log n\right)\right\}$ && $\epsilon$-error, Non-adaptive\\
 & & & $O\left(D \log n + R^2 \log n\right)$ & Non-adaptive, Probabilistic\\
DCP & UB/ UB & $\max\left\{\Omega\left(D \log n\right),\Omega\left(\frac{R}{\log {R}}\log n\right)\right\}$&&$\epsilon$-error, Non-adaptive\\
 & & & $O\left((R+D) \log n\right)$ & Non-adaptive, Probabilistic\\
\hline
\end{tabular}
\label{tab1}
\end{table*}
A summary of known adaptive and non-adaptive pooling designs for GTI and comparison with this work is given below as well as in Table \ref{tab1}\footnote{Throughout this paper, upper bounds on the number of inhibitors and defectives are denoted by $R$ and $D$, and their exact numbers are denoted by $r$ and $d$ respectively.}. 

A lower bound of $\Omega\left((R+D)\log n\right)$ tests was identified for SCP as well as DCP in \cite{FKKM}. The fully adaptive algorithm proposed in \cite{FKKM} is order optimal in the expected number of tests. Non-adaptive pooling designs in the classical group testing framework is known to be related to the design of superimposed codes \cite{DyR}. The connection between identifying the defectives amidst inhibitors and cover-free families, which is a generalization of superimposed codes, was exploited in \cite{DeV} to obtain a better lower bound on the number of tests for DCP than in \cite{FKKM} for the zero-error scenario. Lower bounds for DCP for non-adaptive pooling designs was subsequently studied in \cite{ADB} in the case where the upper bound on the number of inhibitors is given by $R$ and the upper bound on the number of defectives is given by $D$. For this case, an upper bound of $O\left((R+D)^2\log n\right)$ tests was obtained using the notion of $(R,D)$-inhibitory design introduced in \cite{DMTV}. This upper bound is away by a multiplicative factor of $\log {(R+D)}$ from the lower bound of $\Omega\left(\frac{(R+D)^2}{\log(R+D)}\log n\right)$ tests necessary for $(R,D)$-inhibitory designs as shown in \cite{ADB}. However the decoding complexity of the algorithm in \cite{DMTV} is of the order $O(n^{R}(R+D)^2 \log n)$ time units. Chang et al. proposed a non-adaptive pooling design with low decoding complexity for zero-error DCP in \cite{CCF}. For this pooling design, the scaling of the number of tests is the same as in \cite{DMTV} but the decoding complexity is much less, i.e., $O(n(r+d)^2 \log n)$ time units. A non-adaptive pooling design for zero-error SCP was also proposed in \cite{CCF} for the case where the exact number of defectives and inhibitors is known. The design requires $O((r+d)^3 \log n)$ tests and has a decoding complexity of $O(n(r+d)^3 \log n)$ time units. A lower bound on the number of tests given by $\max\left\{\Omega\left(\frac{(r+d)^2}{\log (r+d)}\log n\right), \Omega\left(\frac{r^3}{\log r}\log n\right)\right\}$ was also identified in \cite{CCF}.

It was shown in \cite{ADB} that, when the exact number of defectives is given by $d$ and the upper bound on the number of inhibitors is given by $R$, a lower bound on the number of tests for any $k$-stage adaptive pooling design for DCP with zero-error is given by 
\begin{align} \label{eqn-LB-Adap}
 &\Omega\left(\frac{R^2}{d \log R}\log n+R\log n + d\log n\right), \text{ if $R \geq 2d$},\\
 \nonumber
 &\Omega\left(R\log n+d\log n\right), \text{ if $R < 2d$}.
\end{align} Further, \cite{ADB} reported a trivial two-stage adaptive pooling design that meets the lower bound in (\ref{eqn-LB-Adap}) for $R <2d$ and exceeds the lower bound by a $\log R$ multiplicative factor for $R \geq 2d$. The above lower bound is also met closely by a four-stage adaptive pooling design proposed in \cite{DGV}. For zero-error SCP, when only the upper bounds on the number of defectives and inhibitors are known, a lower bound on the number of tests for any $k$-stage adaptive pooling design was shown to be $\Omega \left(\frac{R^2}{\log R}\log n + D \log n\right)$ \cite{ADB}. This lower bound is met upto a $\log R$ multiplicative factor by the three-stage adaptive pooling design proposed in \cite{DeV}.

In this paper\footnote{It is assumed throughout this paper that $r,d=o(n)$.}, we first propose a probabilistic non-adaptive pooling design for SCP that achieves the trivial lower bound of $\Omega(d\log n)$ tests in the regime $r=O(d)$. In the regime $d=o(r)$, the proposed algorithm is shown to exceed the lower bound by a $\log \frac{r}{d}$ multiplicative factor. Similar matches between the upper bound and lower bound in the $r=O(d)$ regime and gaps between the upper bound and lower bound in the $d=o(r)$ regime are observed for DCP too, as shown in Table \ref{tab1}. Tolerating an error probability that vanishes with the codeword length is known to offer significant gains in the rate of transmission compared to targeting zero-error estimation of messages in communication theory. A similar principle was observed with probabilistic non-adaptive pooling designs for classical group testing where the defectives can be classified in $O(d\log n)$ tests \cite{CJSA}, with an error probability that vanishes with $n$, while it requires $\Omega(\frac{d^2}{\log d}\log n)$ tests if zero-error is insisted upon \cite{book-HD}. In this paper too, a significant gain in the number of tests is observed compared to the number tests required for non-adaptive pooling design for zero-error SCP derived in \cite{CCF} where the knowledge of exact number of inhibitors and defectives was assumed. We also propose a non-adaptive pooling design for the case where knowledge of only upper bounds on the number of inhibitors and defectives is assumed. For this case, for SCP it is shown that the upper bound and the lower bound matches in the $R^2=O(D)$ regime while the upper bound exceeds the lower bound in the $D=o(R^2)$ regime by a $\log {R}$ multiplicative factor. A similar phenomenon is observed for DCP too, as observed from Table \ref{tab1}.


The main contributions and organization of this paper are summarized below.

\begin{itemize}
 \item A probabilistic non-adaptive pooling design for the GTI model with exact knowledge of number of inhibitors and defectives is proposed in Section \ref{sec4}. The number of tests required to guarantee an error probability of $cn^{-\delta}$, $c,\delta>0$, for SCP is shown to be $T=O(d\log n)$ in the $r=O(d)$ regime, and $T=O\left(\frac{r^2}{d}\log n\right)$ in the $d=o(r)$ regime (derived in Section \ref{subsec1}).
 \item An asymptotic lower bound on the number tests for SCP given by $\Omega\left(\frac{r^2}{d \log\frac{r}{d}}\log n\right)$ in the $d=o(r)$ regime for non-adaptive pooling designs is derived in  Section \ref{sec5}. Thus, the number of tests required for the proposed pooling design exceeds the lower bound by a $\log \frac{r}{d}$ multiplicative factor in the $d=o(r)$ regime.
 \item For the case where only upper bounds on the number of inhibitors and defectives is known, a modified non-adaptive pooling design is proposed and lower bound on the number of tests is derived in Section \ref{sec6}. For SCP, the upper bound is shown to match the trivial lower bound of $\Omega(D \log n)$ in the $R^2=O(D)$ regime and exceed the lower bound of $\Omega\left(\frac{R^2}{\log R} \log n\right)$ by a $\log R$ multiplicative factor in the $D=o(R^2)$ regime.
 \item The decoding complexity of the proposed algorithms are given by $O(nT)$.
\end{itemize}

In the next section, we formally introduce the GTI model.

{\em Notation:} The probability of an event $\cal E$ is denoted by $\Pr \{\cal E\}$. The notation $f(n) \approx g(n)$ represents approximation of a function $f(n)$ by $g(n)$. Mathematically, the approximation denotes that for every $\epsilon>0$, there exists $n_0$ such that for all $n>n_0$, $1-\epsilon<\frac{|f(n)|}{|g(n)|}<1+\epsilon$. The notation $\log x$ denotes logarithm to the base two and $\ln$ denotes natural logarithm.

\section{Model} \label{sec2}
The pools are chosen according to a $T \times n$ binary test matrix $M$, where $T$ denotes the number of tests and $n$ denotes the number of items. If the items are indexed from $1$ to $n$, item-$j$ participates in a test $i$ if $m_{ij}=1$, where $m_{ij}$ denotes the $i^{\text{th}}$-row, $j^{\text{th}}$-column entry of the matrix $M$. In other words, a column of the matrix denotes an item and a row denotes a test. The outcome of a test $i$, denoted by $Y_i$, is positive or equal to one iff at least one defective and no inhibitors are present in the test, and negative or zero otherwise. For example, if item-$1$ is a defective, item-$2$ is an inhibitor, and item-$3$ is a normal item, then the outcome vector corresponding to the test matrix
\begin{align*}
 M=\begin{bmatrix}
    1 & 1 & 0\\
    1& 0 &1\\
    1& 0 & 0
   \end{bmatrix}
\end{align*}is given by $\underline{Y}=[0 ~1 ~1]^T$. The goal is to identify all the $d$ defectives and $r$ inhibitors in the population using non-adaptive tests in as minimum a number of tests as possible.

This work is inspired by the noisy-CoMa algorithm \cite{CJSA} which is briefly reviewed in the next section.

\section{Review of Noisy CoMa Algorithm \cite{CJSA}} \label{sec3}
The pooling design and decoding algorithm proposed in this paper is inspired by the noisy column matching (CoMa) algorithm proposed in \cite{CJSA} for the classical group testing framework. We briefly summarize the CoMa algorithm and its noisy version in this section. An instance of classical group testing in the noiseless case and noisy case is given below. The first two columns of $M$ correspond to defective items and the last two columns correspond to normal items.
\begin{align}
\label{eqn-eg_CoMa1}
M=\begin{bmatrix}
\color{blue}{0}&\color{blue}{1} & 0 & 0\\
\color{blue}{1}&\color{blue}{0} & 1 & 0 \\
\color{blue}{0}&\color{blue}{0} & 1 & 0 \\
\color{blue}{1}&\color{blue}{1} & 0 & 1
\end{bmatrix}&\underset{\text{Noiseless}}{\Rightarrow}
\underline{Y}=\begin{bmatrix}
1\\
1\\
0\\
1
\end{bmatrix}\\
\label{eqn-eg_CoMa2}
&\underset{\text{~~Noisy}}\Rightarrow~
\underline{Y}=\begin{bmatrix}
1\\
1\\
0\\
\color{red}{0}
\end{bmatrix}.
\end{align}
%

In the noisy case, the noise flips an outcome with a probability $q$ and the noise is i.i.d. across tests. In the noiseless case, the CoMa algorithm ``matches'' the column of each item to the outcome. If an item participates only in positive outcome tests then, it is declared to be a defective and otherwise, it is declared to be a normal item. Clearly, a defective item participates only in positive outcome tests. But a normal item in any test could be masked by a defective item. Hence, when all the tests in which a normal item participates are masked by at least one defective, the normal item is erroneously declared as a defective. The example in (\ref{eqn-eg_CoMa1}) illustrates this. The normal item-$4$ is masked by item-$1$ in the fourth test and hence, is erroneously declared as a defective. However, item-$3$ is correctly declared to be a normal item because it participates in a negative outcome test, i.e., the third test. It is shown in \cite{CJSA} that when the entries of the test matrix $M$ is chosen according to i.i.d. ${\cal B}(p)$\footnote{Such a pooling design is referred to as an i.i.d. pooling design with parameter $p$.}, where ${\cal B}(p)$ represents the Bernoulli distribution with parameter $p=\frac{1}{d}$, the probability of a normal item to be declared as a defective item vanishes with $n$, given sufficient number of tests, i.e., $O(d \log n)$ tests.

In the noisy case, even the defective items could participate in negative outcome tests on account of noise flipping the outcomes. So, the noisy CoMa algorithm proposes a threshold based algorithm to differentiate between the defectives and the normal items. In the example given in (\ref{eqn-eg_CoMa2}), the defectives can be correctly identified if the criterion for declaring the defectives is relaxed. For example, the criterion for the example in (\ref{eqn-eg_CoMa2}) could be that an item is declared to be a defective if it participates in not more than a single negative outcome test. With this criterion, both the defectives are correctly identified. This idea was generalized by the noisy CoMa algorithm as follows.

Denote the set of tests in which an item-$j$ participates by ${\cal T}_j$ and the set of positive outcome tests in which an item-$j$ participates by ${\cal S}_j$, for $j=1,2,\cdots,n$. Mathematically, ${\cal T}_j$ and ${\cal S}_j$ are defined as follows.
\begin{align} \label{eqn-sets}
& {\cal T}_j \triangleq \{i|m_{ij}=1\},\\
\nonumber
& {\cal S}_j \triangleq \{i|m_{ij}=1 \text{ and } Y_i=1\}.
\end{align}An item-$j$ is declared to be a defective if $|{\cal S}_j| > |{\cal T}_j|[1-q(1+\tau))]$, for some $\tau>0$, and a normal item, otherwise, where $q$ is the probability that an outcome is flipped. It was shown that with an i.i.d. pooling design with $p=\frac{1}{d}$, appropriate choice of $\tau$ and a sufficient number of tests, all the items can be classified with $\epsilon$-error probability.

In the next section, we demonstrate an adaptation of the noisy CoMa algorithm for SCP in the GTI model.

\section{Pooling design and Decoding Algorithm} \label{sec4}
In this section, exact knowledge of the number of inhibitors and defectives given by $r$ and $d$ respectively is assumed. As summarized in the Section \ref{sec1}, the inhibitor model was motivated by ``noisy'' measurements due to presence of inhibitory compounds in blood testing. Hence, it is natural to consider an adaptation of the noisy CoMa algorithm for SCP where, from the view point of the defectives, the inhibitors behave as noise. 
In the noisy classical group testing, the noise is assumed to be independent across tests. A similar effect can be accomplished when an i.i.d. pooling design with parameter $p$ is used. In this section, the ``noise'' parameter $q$ denotes the probability of presence of at least one inhibitor in a test, which is given by $1-(1-p)^r$. Hence, with an i.i.d. pooling design, as far as the defectives are concerned, its positive outcomes are flipped with probability $q$. However, unlike in the noisy classical group testing, the ``noise'' is asymmetric here. To see this, consider a test where a normal item and no defective participates. Here, the outcome is never flipped unlike in the noisy classical group testing framework.

For an i.i.d. pooling design with parameter $p$ which shall be specified later, the decoding algorithm for declaring the inhibitors, normal items and the defectives is specified below as well as represented in Fig. \ref{fig-threshold}.

\begin{enumerate}
 \item If $|{\cal S}_j|=0$, declare item-$j$ to be an inhibitor.
 \item If $1 \leq |{\cal S}_j| \leq \lfloor|{\cal T}_j|[1-q(1+\tau))]\rfloor$, declare item-$j$ to be a normal item.
 \item If $|{\cal S}_j| > |{\cal T}_j|[1-q(1+\tau))]$, declare item-$j$ to be a defective.
\end{enumerate} 
\begin{figure}[htbp] 
\vspace{-3cm} 
\includegraphics[totalheight=11cm,width=10cm]{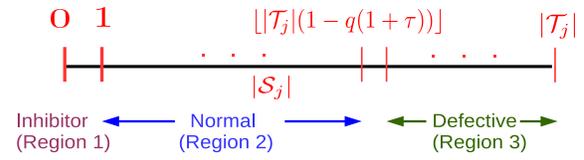}
 \vspace{-5cm} 
\caption{Thresholds for item classification.}
\label{fig-threshold}
\end{figure}
Clearly, all the inhibitors participate only in negative outcome tests (i.e., the inhibitors never fall under region-$2$ or region-$3$). And, one expects a defective to participate in a large fraction of positive outcome tests and a normal item to participate in a relatively few positive outcome tests. This statistical difference is exploited to differentiate the defectives from the normal items. The average number of negative outcome tests in which a defective item-$j$ participates is given by $ |{\cal T}_j| q$. So, a defective item-$j$ is expected to participate in roughly $|{\cal T}_j|(1-q)$ positive outcome tests. The slack parameter $\tau$ ensures that all the defectives are identified with high probability. However, the value of $\tau$ must not be too high because too many normal items would be wrongly identified as defectives. A careful choice of $\tau$ to ensure $\epsilon$-error probability, where $\epsilon=cn^{-\delta}$, is specified in the following sub-section. 

We note that the decoding rule for identifying the defectives is similar to that in the noisy CoMa algorithm. The difference is that the probability $q$ is now dependent on the pooling design parameter. Now, it is not clear what the choice of the parameter must be so that the number of tests required is close to the lower bound. Non-trivial lower bounds for $\epsilon$-error non-adaptive schemes are also not known. In the following sub-section, using a suitable choice of $p$, we show that the number of tests required for guaranteed classification of items with high probability scales as $T=O(d\log n)$ and $T=O(\frac{r^2}{d}\log n)$ in the $r=O(d)$ and $d=o(r)$ regimes respectively. And, in the following section, we obtain a non-trivial lower bound for $\epsilon$-error non-adaptive pooling design.

\subsection{Error Analysis of the Proposed Algorithm} \label{subsec1}

The error analysis of the proposed algorithm follows similar techniques used in the analysis of CoMa algorithm in \cite{CJSA}. In the proposed algorithm, there are three possible error events for a given item that participates in at least one of the tests. 

\begin{enumerate}
\item  A defective might be wrongly identified as a normal item or an inhibitor, i.e., the defective falls under either region-$1$ or region-$2$ in  Fig. \ref{fig-threshold}. 
\item A normal item might be wrongly identified as a defective, i.e., the normal item falls under region-$3$ in  Fig. \ref{fig-threshold}.
\item A normal item might be wrongly identified as an inhibitor, i.e., the normal item falls under region-$1$ in  Fig. \ref{fig-threshold}.
\end{enumerate}
The other error events include non-participation of items in any of the tests. Clearly, an inhibitor cannot be wrongly identified as a normal item or a defective in the proposed decoding algorithm. Now, denote the set of defectives, normal items, and inhibitors by ${\cal D}, {\cal N}$, and ${\cal I}$ respectively. Define the following error events which are related to the error events described above.
\begin{align*}
\nonumber
&{\cal E}^{(j)}_1 \triangleq \text{Item-$j$, $j\in {\cal D}$, does not appear in any of the tests or}\\
&~\hspace{1cm}\text{it is declared as a normal item or an inhibitor.}\\
\nonumber
&{\cal E}^{(j)}_2 \triangleq \text{Item-$j$, $j\in {\cal N}$, does not appear in any of the tests or}\\
&~\hspace{1cm}\text{it is declared as a defective.}\\
&{\cal E}^{(j)}_3 \triangleq \text{Item-$j$, $j\in {\cal N}$, does not appear in any of the tests or}\\
&~\hspace{1cm}\text{it is declared as an inhibitor.}\\ 
\nonumber
&{\cal E}^{(j)}_4 \triangleq \text{Item-$j$, $j\in {\cal I}$, does not appear in any of the tests.}
\end{align*}Clearly, the total error probability is equal to the probability of union of the error events defined above. We now evaluate the probability of each of these error events and show that, when $r=O(d)$ and $d=o(r)$, $T=O(d\log n)$ and $T=O(\frac{r^2}{d}\log n)$ tests are sufficient respectively to guarantee that the total error probability decays as $c n^{-\delta}$ for some positive constant $\delta$ and $c=4$ (corresponding to the union bound of the four error events define above).

Let $T=\beta \log n$. Now, we have
\begin{align}
\nonumber
&\text{Pr}\left\{\underset{j\in {\cal D}}{\bigcup} {\cal E}^{(j)}_1\right\} \leq d\text{ Pr} \left\{ |{\cal S}_j| \leq |{\cal T}_j|[1-q(1+\tau))]|j \in {\cal D}\right\}\\
\nonumber
&= d \sum_{t=0}^{T} {T \choose t} p^t (1-p)^{T-t} \sum_{v=tq(1+\tau)}^{t} {t \choose v} q^v (1-q)^{t-v}\\
\nonumber
&\overset{(a)}{\leq} d \sum_{t=0}^{T} {T \choose t} p^t e^{-2t{(q\tau)}^2} (1-p)^{T-t}\\
\nonumber
&\overset{(b)}{=} d \left[1-p+pe^{-2 (q\tau)^2}\right]^{\beta \log n}\\
\nonumber
&\overset{(c)}{\leq} d\exp\left\{-\beta p \log n \left(1-e^{-2 (q\tau)^2}\right)\right\} \leq n^{-\delta}\\
\nonumber
&\overset{(d)}{\Rightarrow}  d\exp\left\{-\beta p \log n \left(1-e^{-2}\right)(q\tau)^2\right\} \leq n^{-\delta}\\
\label{eqn-beta1}
&\Rightarrow \beta \geq \frac{\left(\frac{\ln d}{\ln n}+\delta\right)\ln 2}{p(1-e^{-2})(q\tau)^2},
\end{align}where $(a)$ follows from Chernoff-Hoeffding bound \cite{Hof}, $(b)$ follows from binomial expansion, $(c)$ follows from the fact that $1-c \leq e^{-c}$, and $(d)$ suffices to ensure the upper bound of $n^{-\delta}$ since $\left(1-e^{-2 (q\tau)^2}\right) \geq \left(1-e^{-2}\right)(q\tau)^2$, for $0 < q\tau < 1$.

Let $a$ denote the probability that a test outcome is positive given that a normal item-$j$ is present in the test. This is given by
\begin{align}
\label{eqn-define-a}
a=(1-p)^r[1-(1-p)^d]. 
\end{align}
We now have

\begin{align*}
&\text{Pr}\left\{\underset{j\in {\cal N}}{\bigcup} {\cal E}^{(j)}_2\right\}\\
&\leq (n-d-r)\text{ Pr} \left\{ |{\cal S}_j| > |{\cal T}_j|[1-q(1+\tau))]|j \in {\cal N}\right\}\\
&= (n-d-r) \sum_{t=0}^{T} {T \choose t} p^t (1-p)^{T-t} \\\nonumber &\hspace{3cm} \times\sum_{v=t\left(1-q(1+\tau)\right)}^{t} {t \choose v} a^v (1-a)^{t-v}\\
&{=} (n-d-r) \sum_{t=0}^{T} {T \choose t} p^t (1-p)^{T-t} \\\nonumber &\hspace{2.5cm} \times \sum_{v=at+t\left(1-q(1+\tau)-a\right)}^{t} {t \choose v} a^v (1-a)^{t-v}.
\end{align*}Now, following similar steps as in obtaining (\ref{eqn-beta1}) and assuming that $0<(1-q(1+\tau)-a)<1$\footnote{This condition is necessary for application of Chernoff-Hoeffding bound.}, we have
\begin{align}
\label{eqn-beta2}
&\beta \geq \frac{\left(\frac{\ln n-d-r}{\ln n}+\delta\right)\ln 2}{p(1-e^{-2})\left(1-q(1+\tau)-a\right)^2} .
\end{align}

The constraint on $\beta$ corresponding to the third error event is obtained as follows.
\begin{align}
\nonumber 
 &\text{Pr}\left\{\underset{j\in {\cal N}}{\bigcup} {\cal E}^{(j)}_3\right\}\leq (n-d-r)\text{ Pr} \left\{ |{\cal S}_j| = 0|j \in {\cal N}\right\}\\
\nonumber
&= (n-d-r) \sum_{t=0}^{T} {T \choose t} p^t (1-p)^{T-t} (1-a)^t\\
\nonumber
&= (n-d-r)(1-ap)^T \\\nonumber & \leq (n-d-r)\exp\{-ap\beta \log n\} \leq n^{-\delta}\\
\label{eqn-beta3}
&\Rightarrow \beta \geq \frac{\left(\frac{\ln n-d-r}{\ln n}+\delta\right)\ln 2}{p(1-p)^r\left(1-(1-p)^d\right)}.
\end{align}

Now, from the fourth error event we obtain
\begin{align}
 \nonumber
  &\text{Pr}\left\{\underset{j\in {\cal I}}{\bigcup} {\cal E}^{(j)}_4\right\} \leq r (1-p)^T \leq n^{-\delta}\\
 \label{eqn-beta4}
 \Rightarrow &~ \beta \geq  \frac{\left(\frac{\ln r}{\ln n}+\delta\right)\ln 2}{p}.
\end{align}

Note that the term $\frac{1}{q\tau}$, which appears in (\ref{eqn-beta1}), is a decreasing function of $\tau$ whereas $\frac{1}{\left(1-q(1+\tau)-a\right)}$, which appears in (\ref{eqn-beta2}), is an increasing function of $\tau$. Therefore, the minimum of the two quantities is maximized when 
\begin{align}
\label{eqn-tau}
 \frac{1}{q\tau} = \frac{1}{\left(1-q(1+\tau)-a\right)} \Rightarrow \tau = \frac{1-q-a}{2q}.
\end{align}Note that the above value of $\tau$ satisfies the conditions $0<q\tau<1$ and $0<(1-q(1+\tau)-a)<1$ which were assumed in deriving (\ref{eqn-beta1}) and (\ref{eqn-beta2}) respectively. Substituting this value of $\tau$, we have

{\small\begin{align*}
\frac{1}{(q\tau)^2} = \frac{1}{\left(1-q(1+\tau)-a\right)^2} = \frac{4}{(1-q-a)^2} =\frac{4}{(1-p)^{2(r+d)}}.
\end{align*}}Therefore, from (\ref{eqn-beta1})-(\ref{eqn-beta4}) we have

{\small\begin{align}\nonumber
\beta \geq \max & \left \{\frac{4\left(\frac{\ln d}{\ln n}+\delta\right)\ln 2}{p(1-p)^{2(r+d)}(1-e^{-2})}, \frac{4\left(\frac{\ln n-d-r}{\ln n}+\delta\right)\ln 2}{p(1-p)^{2(r+d)}(1-e^{-2})},\right.\\ \label{eqn-no_of_tests_p} & ~~\left.\frac{\left(\frac{\ln n-d-r}{\ln n}+\delta\right)\ln 2}{p(1-p)^r\left(1-(1-p)^d\right)}, \frac{\left(\frac{\ln r}{\ln n}+\delta\right)\ln 2}{p}\right \}.
\end{align}}Since $1-p \leq e^{-p}$, if $p$ does not scale inversely with respect to $r+d$, the first two terms above would scale exponentially with $r$. Optimizing the denominators of the first two terms above, we have 
\begin{align*}
p=\frac{1}{2(r+d)+1}. 
\end{align*}At large values of $r+d$, using the approximations $p \approx \frac{1}{2(r+d)}$ and $1-p \approx e^{\frac{-1}{2(r+d)}}$, we obtain (\ref{eqn-no-of-tests}) (given at the top of the next page).

\begin{figure*}
\begin{align} \label{eqn-no-of-tests}
\beta \geq \max & \left \{\frac{8e(r+d)\left(\frac{\ln d}{\ln n}+\delta\right)\ln 2}{(1-e^{-2})}, \frac{8e(r+d)\left(\frac{\ln n-d-r}{\ln n}+\delta\right)\ln 2}{(1-e^{-2})},\frac{2(r+d)\left(\frac{\ln n-d-r}{\ln n}+\delta\right)\ln 2}{e^{\frac{-r}{2(r+d)}}\left(1-e^{\frac{-d}{2(r+d)}}\right)}, {2(r+d)\left(\frac{\ln r}{\ln n}+\delta\right)\ln 2}\right \}
\end{align}\hrule
\end{figure*}

The scaling of the number of tests in the proposed algorithm in the two possible scenarios $r=O(d)$ and $d=o(r)$ is evaluated as given below.

\begin{enumerate}
 \item For $r=O(d)$, i.e., $r \leq cd$ for some constant $c>0$ and for all $n>n_0$, the third term in (\ref{eqn-no-of-tests}) can be upper bounded by
 
 \begin{align*}
\frac{2(r+d)\left(\frac{\ln n-d-r}{\ln n}+\delta\right)\ln 2}{e^{\frac{-c}{2(c+1)}}\left(1-e^{\frac{-1}{2(c+1)}}\right)}.
 \end{align*}Since the rest of the terms in (\ref{eqn-no-of-tests}) scale as $r+d$ and recalling that $T=\beta \log n$, it is easily seen that the number of tests scales as $T=O\left(d\log n\right)$.
 
 \item For $d=o(r)$, using the approximation $1-e^{\frac{-d}{2(r+d)}} \approx \frac{d}{2r}$, the third term in (\ref{eqn-no-of-tests}) is approximated as
  \begin{align*}
  \frac{4e~r(r+d)\left(\frac{\ln n-d-r}{\ln n}+\delta\right)\ln 2}{d}.
 \end{align*}Hence, the number of tests scales as $T=O\left(\frac{r^2}{d}\log n\right)$ in the $d=o(r)$ regime.
\end{enumerate} 
 \begin{remark} \label{rem1}
 Note that the proposed algorithm identifies all the types of items. However, if the objective is to identify only the defectives, it can be done in $T=O\left((r+d)\log n\right)$ tests using the proposed algorithm. This is because only the first and the second terms in (\ref{eqn-no-of-tests}) matter as the third and the fourth terms correspond to the events of wrongly identifying normal items as inhibitors and inhibitors not appearing in any of the tests, respectively. Also, note that the time complexity of the proposed algorithm is given by $O(nT)$.
 \end{remark}
 
\begin{remark}
Using (\ref{eqn-tau}), we observe that the threshold for differentiating the defectives from the normal items is given by
\begin{align*}
 {\cal T}_j(1-q(1+\tau)) = {\cal T}_j\frac{b+a}{2},
\end{align*}where $b=1-q$ represents the probability of a positive outcome given that a defective item-$j$ is present in the test\footnote{The value of $\tau$ could also be chosen by optimizing the full expressions in (\ref{eqn-beta1}) and (\ref{eqn-beta2}) instead of optimizing only their denominators. However, since we are interested only in order optimality, optimizing only their denominators suffices. Moreover, the resulting value of $\tau$ helps formulate a thumb rule for i.i.d. pooling design. }. Note that the chosen threshold can be seen as the average of the statistics of the two different types of items. In general, if the threshold for differentiating two different types of items is fixed in this manner, where $b>a$, the number of tests required to guarantee $\epsilon$-error probability scales inversely as
\begin{align*}
 T \propto \frac{1}{p(b-a)^2}.
\end{align*}We note that the first two terms in (\ref{eqn-no_of_tests_p}) take this form. This is expected because each item appears in a test with probability $p$ and the number of tests must be inversely proportional to the statistical difference between the two different types of items. This observation might serve as a useful thumb rule in predicting the scaling of the number of tests required for i.i.d. pooling designs. For example, this thumb rule could be useful when i.i.d. pooling designs are used for stochastic threshold group testing which exploits statistical difference between a defective and a normal item in classifying the items, when the number of defectives in a pool is known to be exactly equal to the lower threshold \cite{CCBJS}.
\end{remark}

Now, clearly, the proposed pooling design requires order optimal number of tests when $r=O(d)$. But, for $d=o(r)$, the lower bound is unknown in the $\epsilon$-error case. We now present a lower bound tighter than $\Omega(r\log n)$ in the $d=o(r)$ regime for non-adaptive pooling designs in the next section. 

\section{Lower Bound for Non-Adaptive Pooling Design} \label{sec5}
In this section, we show that the  number of tests required in the proposed non-adaptive scheme exceeds the derived lower bound by a $\log \frac{r}{d}$ multiplicative factor. The number of choices for the defectives and inhibitors is given by ${n \choose d}{n-d \choose r}$. A lower bound on the number of tests is now given by 
\begin{align} \label{eqn-LB}
T \geq  \frac{\log {n \choose d}{n-d \choose r}(1-P_e)-H_2(P_e)}{\underset{g}{\max}~ H(Y)},
\end{align}where $g$ represents the size of a pool, $H(Y)$ represents the outcome entropy, $P_e$ denotes the average error probability\footnote{The average error probability is defined as $\Pr\{\hat{X} \neq X\}$, where $X$ denotes the actual classification of the $n$ items and $\hat{X}$ denotes the estimated classification of the $n$ items. The probability is averaged over the randomness in the actual classification of the items.}, and $H_2(P_e)$ denotes the binary entropy. Standard information theoretic arguments along with Fano's inequality are used to obtain (\ref{eqn-LB}). The details can be found in \cite{CJSA} in the classical group testing framework. The average error probability $P_e$ is assumed to vanish to zero with increasing $n$. It is also assumed that the number of inhibitors $r$ grows with the number of items $n$ and $r,d=o(n)$. 

Now, let $p_Y$ represent the probability of a positive outcome. If $p_Y \leq \frac{1}{2}$ then, maximizing the outcome entropy is equivalent to maximizing $p_Y$. In the regime $d=o(r)$, for all sufficiently large $n \geq n_0$, we show that the optimum pool size that maximizes $p_Y$ is given by
$$
g_{opt}\in \left\{
\begin{array}{ll}
\left[\lfloor\frac{n}{r}\rfloor,\lceil\frac{n}{r}\rceil\right]_{\mathbb Z},& ~\text{if $\frac{n}{r}$ is not an integer} \\
\left[\frac{n}{r}-1,\frac{n}{r}+1\right]_{\mathbb Z}, & ~\text{otherwise,}
\end{array}
\right.
$$where $[a,b]_{\mathbb Z}$ represents integers between $a$ and $b$ including the end points. In other words, the optimum pool size can be expressed as $g_{opt}=\frac{n}{r} + \alpha$, where $-1\leq \alpha \leq 1$. We show that the maximum value of $p_Y$ is approximately given by $p_{Y_{max}} \approx \frac{d}{r}$ which is less than half for large $n$. Hence, the entropy $H(Y)$ is also maximized at $g_{opt}=\frac{n}{r} + \alpha$, for large $n$.

The probability of positive outcome for a pool size of $g \leq n-d-r$ is given by 
\begin{align*}
 p_Y(g)=\frac{{n-r \choose g}-{n-d-r \choose g}}{{n \choose g}}.
\end{align*}We now prove that $p_Y(g)$ is an increasing function for $g < g_0$ and a decreasing function for $g > g_1$. This implies that the optimum value of $g$, subject to the constraint that $g \leq n-d-r$, lies between $g_0$ and $g_1$. In fact, it is shown that the real numbers $g_0$ and $g_1$ converge to $\frac{n}{r}$ with increasing $n$. It is shown later that the global optimum pool size falls in the interval $g \leq n-d-r$, for all sufficiently large $n$. Towards that end, the probability $p_Y(g)$, for $g \leq n-d-r$, is re-written as in (\ref{eqn-p_y}) (given at the top of the next page).

Unless mentioned otherwise, hereafter, we assume that $g \leq n-d-r$.

\begin{figure*}
\begin{align}\nonumber
p_Y(g) & = \frac{(n-r)(n-r-1)\cdots (n-r-g+1)}{n (n-1) \cdots (n-g+1)} - \frac{(n-d-r)(n-d-r-1)\cdots (n-d-r-g+1)}{n (n-1) \cdots (n-g+1)}\\
\label{eqn-p_y}
& = \prod_{i=0}^{g-1} \left(1-\frac{r}{n-i}\right) - \prod_{i=0}^{g-1} \left(1-\frac{r+d}{n-i}\right)
\end{align}
\hrule
\end{figure*}

\begin{lemma} \label{lem_dec_g1}
The probability $p_Y(g)$ is a decreasing function of $g$ for $g > g_1$, where
 $g_1=\frac{\ln \left(1+\frac{d}{r}\right)}{ \ln \left(1+\frac{d}{n-d-r}\right)}$.
\end{lemma}
\begin{proof}
It is shown below that  $p_Y(g) > p_Y(g+1)$, for $g > g_1$.
\begin{align} \nonumber
&~\hspace{0.5cm}p_Y(g) > p_Y(g+1) \\
\nonumber&\Leftrightarrow \prod_{i=0}^{g-1} \left(1-\frac{r}{n-i}\right) \left(1-\left(1-\frac{r}{n-g}\right)\right) \\ \nonumber & \hspace{2cm} > \prod_{i=0}^{g-1} \left(1-\frac{r+d}{n-i}\right) \left(1-\left(1-\frac{r+d}{n-g}\right)\right)\\
\label{eqn-g1}
&\Leftrightarrow \prod_{i=0}^{g-1} \left(1+\frac{d}{n-d-r-i}\right) > 1+\frac{d}{r}\\
\nonumber
&\Leftarrow \left(1+\frac{d}{n-d-r}\right)^g >  1+\frac{d}{r}\\
\label{eqn-g1_approx}
& \Leftrightarrow g > g_1 = \frac{\ln \left(1+\frac{d}{r}\right)}{ \ln \left(1+\frac{d}{n-d-r}\right)} \approx \frac{\frac{d}{r}}{\frac{d}{n-d-r}} \approx \frac{n}{r},
\end{align}where the approximations follow from the fact that $d=o(r)$ and $r=o(n)$.
\end{proof} 

On account of the approximation (\ref{eqn-g1_approx}), for sufficiently large $n$, Lemma \ref{lem_dec_g1} implies that $p_Y(g)$ is a decreasing function of $g$ for $g \geq \lceil \frac{n}{r} \rceil$ if $\frac{n}{r}$ is not an integer, and for $g \geq \frac{n}{r}+1$ if $\frac{n}{r}$ is an integer.
%

\begin{lemma} \label{lem_dec_g0}
The probability $p_Y(g)$ is an increasing function of $g$ for $g < g_0$, where
 $g_0=\frac{d+\left(n-d-r+2\right)\ln \left(1+\frac{d}{r}\right)}{ d+ \ln \left(1+\frac{d}{r}\right)}$.
\end{lemma}
\begin{proof}
It is shown below that  $p_Y(g-1) < p_Y(g)$, for $g < g_0$. Following similar steps as in obtaining (\ref{eqn-g1}), we have
\begin{align*}
&p_Y(g-1) < p_Y(g) \Leftrightarrow \prod_{i=0}^{g-2} \left(1+\frac{d}{n-d-r-i}\right) < 1+\frac{d}{r}\\
\nonumber
&\Leftarrow \left(1+\frac{d}{n-d-r-g+2}\right)^{g-1} < 1+\frac{d}{r}\\
& \Leftarrow e^{\frac{(g-1)d}{n-d-r-g+2}} < 1+\frac{d}{r}\\
& \Leftrightarrow g < g_0=\frac{d+\left(n-d-r+2\right)\ln \left(1+\frac{d}{r}\right)}{ d+ \ln \left(1+\frac{d}{r}\right)} \approx \frac{n}{r}
\end{align*}
\end{proof}

Hence, for sufficiently large $n$, Lemma \ref{lem_dec_g0} implies that $p_Y(g)$ is an increasing function of $g$ for $g \leq \lfloor \frac{n}{r} \rfloor$ if $\frac{n}{r}$ is not an integer, and for $g \leq \frac{n}{r}-1 $ if $\frac{n}{r}$ is an integer. Hence, for sufficiently large $n$, the optimum pool size is given by

$$
g_{opt}\in \left\{
\begin{array}{ll}
\left[\lfloor\frac{n}{r}\rfloor,\lceil\frac{n}{r}\rceil\right]_{\mathbb Z},& ~\text{if $\frac{n}{r}$ is not an integer} \\
\left[\frac{n}{r}-1,\frac{n}{r}+1\right]_{\mathbb Z}, & ~\text{otherwise.}
\end{array}
\right.
$$In other words, for $-1\leq \alpha \leq 1$, we have
\begin{align}\label{eqn-g_opt}
g_{opt}=\frac{n}{r} + \alpha=\frac{n}{r}\left(1+\alpha \frac{r}{n}\right). 
\end{align}
We now prove the following asymptotic lower bound.

\begin{figure*}
\begin{align}\label{eqn-thm_LB_1}
p_Y(g_{opt}) &\approx \frac{(n-r-g_{opt})^{-(n-r-g_{opt}+\frac{1}{2})}(n-r)^{n-r+\frac{1}{2}}-(n-d-r-g_{opt})^{-(n-d-r-g_{opt}+\frac{1}{2})}(n-d-r)^{n-d-r+\frac{1}{2}}}{(n-g_{opt})^{-(n-g_{opt}+\frac{1}{2})}n^{n+\frac{1}{2}}}\\
\nonumber
&=\left(1-\frac{g_{opt}}{n}\right)^r \left(1-\frac{r}{n-g_{opt}}\right)^{-\left(n-r-g_{opt}+\frac{1}{2}\right)}\left(1-\frac{r}{n}\right)^{n-r+\frac{1}{2}}\\
\nonumber
&\hspace{3cm}- \left(1-\frac{g_{opt}}{n}\right)^{r+d} \left(1-\frac{r+d}{n-g_{opt}}\right)^{-\left(n-r-d-g_{opt}+\frac{1}{2}\right)}\left(1-\frac{r+d}{n}\right)^{n-r-d+\frac{1}{2}}\\
\nonumber
&\approx e^{-\frac{rg_{opt}}{n}}e^{r-\frac{r^2}{n-g_{opt}}+\frac{r}{2(n-g_{opt})}}e^{-r+\frac{r^2}{n}-\frac{r}{2n}}-e^{-\frac{(r+d)g_{opt}}{n}}e^{(r+d)-\frac{(r+d)^2}{n-g_{opt}}+\frac{(r+d)}{2(n-g_{opt})}}e^{-(r+d)+\frac{(r+d)^2}{n}-\frac{(r+d)}{2n}}\\
\nonumber
&=e^{-\frac{rg_{opt}}{n}}e^{-\frac{r^2g_{opt}}{n(n-g_{opt})}}e^{\frac{rg_{opt}}{2n(n-g_{opt})}}-e^{-\frac{(r+d)g_{opt}}{n}}e^{-\frac{(r+d)^2g_{opt}}{n(n-g_{opt})}}e^{\frac{(r+d)g_{opt}}{2n(n-g_{opt})}}\\
\nonumber
&=e^{-(1+\alpha \frac{r}{n})}e^{-\frac{r(1+\alpha \frac{r}{n})}{n(1-\frac{g_{opt}}{n})}}e^{\frac{(1+\alpha \frac{r}{n})}{2n(1-\frac{g_{opt}}{n})}}\left(1-e^{-\frac{d}{r}(1+\alpha \frac{r}{n})}e^{-\frac{r\left(\frac{d^2}{r^2}+\frac{2d}{r}\right)\left(1+\alpha \frac{r}{n}\right)}{n(1-\frac{g_{opt}}{n})}}e^{\frac{d(1+\alpha \frac{r}{n})}{2nr(1-\frac{g_{opt}}{n})}}\right)\\
\nonumber
&\approx e^{-(1+\alpha \frac{r}{n})}e^{-\frac{r(1+\alpha \frac{r}{n})}{n(1-\frac{g_{opt}}{n})}}e^{\frac{(1+\alpha \frac{r}{n})}{2n(1-\frac{g_{opt}}{n})}}\left(\frac{d(1+\alpha \frac{r}{n})}{r}+\frac{\left(\frac{d^2}{r}+{2d}\right)\left(1+\alpha \frac{r}{n}\right)}{n(1-\frac{g_{opt}}{n})}-\frac{d(1+\alpha \frac{r}{n})}{2nr(1-\frac{g_{opt}}{n})}\right)\\
\label{eqn-py_approx_final}
&\approx\frac{d}{re}.
\end{align}
\hrule
\end{figure*}
\begin{theorem} \label{thm_LB}
An asymptotic lower bound on the number of tests required for non-adaptive pooling designs for SCP is given by $\Omega\left(\frac{r^2}{d \log \frac{r}{d}}\log n\right)$, in the $d=o(r), r=o(n)$ regime.
\end{theorem}
\begin{proof}
Using the fact that $g_{opt}=o(n)$ and using Stirling's approximation for factorial functions in $p_Y(g_{opt})$, we have (\ref{eqn-thm_LB_1})-(\ref{eqn-py_approx_final}) (given at the top of the next page). For $n-d-r<g\leq n-r$, the positive outcome probability is given by
\begin{align*}
&p_Y(g)=\frac{{n-r \choose g}}{{n \choose g}}= \prod_{i=0}^{g-1} \left(1-\frac{r}{n-i}\right)\\
&\leq \left(1-\frac{r}{n}\right)^g \leq e^{-\frac{rg}{n}} \leq e^{-\frac{r(n-d-r)}{n}} \approx e^{-r}.
\end{align*}Also, $p_Y(g)=0$ for $g > n-r$. Thus, for all sufficiently large $n$, $p_Y(g_{opt})>p_Y(g)$, for all $g>n-d-r$.
 
Since $\underset{g}{\max}~H(Y)=-p_Y(g_{opt})\log p_Y(g_{opt}) -(1-p_Y(g_{opt}))\log (1-p_Y(g_{opt}))$, substituting (\ref{eqn-py_approx_final}) in (\ref{eqn-LB}) with $d=o(r)$ and noting that ${n \choose d}{n-d \choose r}\geq(\frac{n}{d})^d (\frac{n-d}{r})^r$, we have
\begin{align*}
 T=\Omega\left(\frac{r^2}{d\log{\frac{r}{d}}}\log n\right).
\end{align*}It is observed that the ratio notion of approximation used suffices because one is interested in order bound on the number of tests.
\end{proof}

Hence, the proposed pooling design and decoding algorithm for SCP exceeds the lower bound by a multiplicative factor of $\log \frac{r}{d}$ tests in the $d=o(r)$ regime. The following result for DCP follows from the proof of the above theorem.

\begin{corollary}\label{cor1}
An asymptotic lower bound on the number of tests required for non-adaptive pooling designs for DCP is given by $\Omega\left(\frac{r}{ \log \frac{r}{d}}\log n\right)$, in the $d=o(r), r=o(n)$ regime. 
\end{corollary}
\begin{proof}
 For DCP, the combinatorial term in the numerator in (\ref{eqn-LB}) is given by ${n \choose d}$ instead of ${n \choose d}{n-d \choose r}$. Since for DCP, only the numerator of (\ref{eqn-LB}) changes w.r.t. SCP, using $\underset{g}{\max} ~H(Y)$ derived in the proof of Theorem \ref{thm_LB} we have the lower bound in the $d=o(r)$ regime to be
 \begin{align*}
  \Omega\left(\frac{d}{\frac{d}{r}\log\frac{r}{d}}\log n\right)=\Omega\left(\frac{r }{\log\frac{r}{d}}\log n\right).
 \end{align*}
\end{proof}

As noted in Remark \ref{rem1}, the proposed non-adaptive pooling design requires $O\left((r+d)\log n\right)$ tests for DCP. Thus, in the $r=O(d)$ regime, the proposed pooling design is order optimal in the number of tests whereas in the $d=o(r)$ regime, it exceeds the lower bound by a $\log \frac{r}{d}$ multiplicative factor.

In the next section, we extend the proposed non-adaptive pooling design and decoding algorithm to the case where only upper bounds on the number of defectives and inhibitors are given. We also exploit the lower bound obtained in this section to obtain a lower bound for the problem discussed in the next section.

\section{GTI with Knowledge of Upper Bounds on Number of Defectives and Inhibitors}\label{sec6}

In this section, it is assumed that only upper bounds on the number of defectives $D$ and the number of inhibitors $R$ are known, with $R,D=o(n)$. It is also assumed that at least one defective is present in the population. Otherwise, there is no way the inhibitors and the normal items can be distinguished. The goal here is to identify all the inhibitors and the defectives with vanishing error probability for any $(r,d)$ inhibitor-defective combination, where $r$ and $d$ denote the actual number of inhibitors and defectives. In other words, if $\underline{X}$ denotes the actual $n \times 1$ input vector that denotes the type of each item and $\underline{\hat{X}}$ denotes the estimated $n \times 1$ input vector, the challenge is to propose a non-adaptive pooling design and decoding algorithm so that 
\begin{align} \label{eqn-ob_UB}
\underset{r\in[0,R],d\in[1,D]}{\max}\Pr\{\underline{\hat{X}}\neq {\underline X}\} \leq cn^{-\delta}, 
\end{align}for some constant $c$ and $\delta>0$. For the lower bound, the random variables $r$ and $d$ are assumed to be uniformly distributed over the intervals $[0,R]$ and $[1,D]$ respectively. It is also assumed that $R\underset{n\rightarrow \infty}{\longrightarrow} \infty$.

For this set-up, we modify the non-adaptive pooling design and decoding algorithm proposed in Section \ref{sec4} for SCP and also utilize the lower bound derived in the previous section to obtain a lower bound for this scenario. 

\subsection{Modified Non-Adaptive Pooling Design and Decoding Algorithm}

To solve the SCP for this GTI scenario, we make use of two i.i.d. pooling designs, a $T_1 \times n$ test matrix $M_1$ chosen according to i.i.d. ${\cal B}(p_1)$ and a $T_2 \times n$ test matrix $M_2$ chosen according to i.i.d. ${\cal B}(p_2)$. The outcomes corresponding to the test matrix $M_1$ (denoted by $\underline{Y_1}$) are used to identify the defectives and the outcomes corresponding to the test matrix $M_2$ (denoted by $\underline{Y_2}$) are used to identify the normal items and inhibitors. So, here the effective test matrix is given by $M=[M^T_1 M^T_2]^T$.

For values of the i.i.d. parameters $p_1$ and $p_2$ to be specified later, the decoding algorithm for declaring the defectives, inhibitors and the normal items  is specified below as well as represented in Fig. \ref{fig-threshold_UB}.

\begin{enumerate}
 \item {\em Stage $1$}: Considering the outcome vector $\underline{Y_1}$, if $ |{\cal S}_j (\underline{Y_1})| > \lfloor|{\cal T}_j(\underline{Y_1})|[1-q_R(1+\tau))]\rfloor$, declare item-$j$ to be a defective.
 \item {\em Stage $2$}: For $j \notin \hat{{\cal D}}$ and considering the outcome vector $\underline{Y_2}$,
 \begin{itemize}
  \item if $|{\cal S}_j(\underline{Y_2})| =0$, declare item-$j$ to be an inhibitor.
  \item if $|{\cal S}_j(\underline{Y_2})| \geq 1$, declare item-$j$ to be a normal item.
 \end{itemize}
\end{enumerate} 
\begin{figure}[htbp] 
\vspace{-1cm} 
\includegraphics[totalheight=11cm,width=10cm]{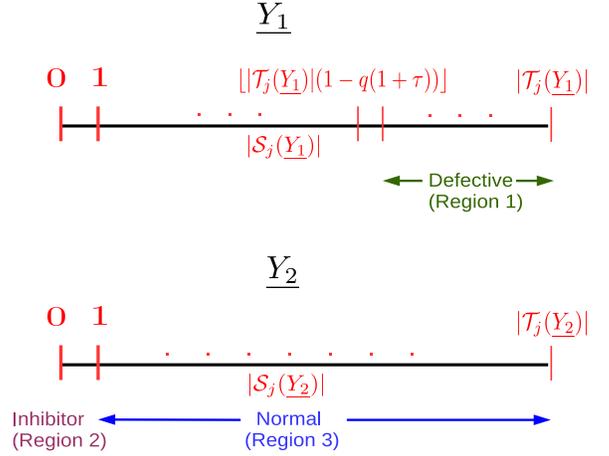}
 \vspace{-3cm} 
\caption{Thresholds for item classification from the outcomes $\underline{Y_1}$ and $\underline{Y_2}$. The second stage of the decoding algorithm that identifies the inhibitors and the normal items from $\underline{Y_2}$ involves only items $j \notin \hat{\cal D}$.}
\label{fig-threshold_UB}
\end{figure}The notation $\hat{{\cal D}}$ denotes the set of items declared to be defectives in Stage $1$. The set of tests corresponding to the test matrix $M_k$, for $k=1,2$, in which an item-$j$ participates is denoted by ${\cal T}_j(Y_k)$ and the set of positive outcome tests corresponding to the test matrix $M_k$ in which an item-$j$ participates by ${\cal S}_j(Y_k)$, for $j=1,2,\cdots,n$. The term $q_R$ denotes the worst case probability of a negative outcome given that a defective is present in a test, i.e., $q_R=1-(1-p_1)^R$. We account for the worst case in order to achieve vanishing error probability for the worst case combination of number of inhibitors and defectives as given in (\ref{eqn-ob_UB}). However, this over-compensation for classifying the defectives might significantly affect the probability of the error event of declaring normal items to be defectives. But it is shown in the following error-analysis that there is no penalty paid for this over-compensation. Note that here, unlike in the pooling design proposed in Section \ref{sec4}, the entries of the effective test matrix $M$ are independent but not identically distributed. The reason why a single i.i.d. test matrix $M$ would be sub-optimal in terms of number of tests will be clear from the error analysis.

\subsubsection{Error Analysis}
As in Section \ref{subsec1}, we enumerate the error events for the proposed algorithm and then find the number of tests required to guarantee that each of these error events vanish with $n$. Three possible error events can occur as given below.
\begin{enumerate}
\item  A defective is not identified as a defective in Stage-$1$ of the decoding algorithm, i.e., the defective does not fall under region-$1$ in Fig. \ref{fig-threshold_UB}. 
\item A normal item might be wrongly identified as a defective in Stage-$1$ of the decoding algorithm, i.e., the normal item falls under region-$1$ in  Fig. \ref{fig-threshold_UB}.
\item A normal item might be wrongly identified as an inhibitor in Stage-$2$ of the decoding algorithm, i.e., the normal item falls under region-$2$ in  Fig. \ref{fig-threshold_UB}.
\end{enumerate}
The other error events include non-participation of items in any of the tests. As before, an inhibitor is never identified as a normal item or a defective. Denoting the set of defectives, normal items, and inhibitors by ${\cal D}, {\cal N}$, and ${\cal I}$ respectively and using the same definition for the error events ${\cal E}^{(j)}_i$ as in Section \ref{subsec1}, their probabilities, for $(r,d)$ being the actual number of inhibitors and defectives, are evaluated as follows.

With $T_1=\beta_1 \log n$ and $q=1-(1-p_1)^r$, we have

{\small\begin{align}
\nonumber
&\text{Pr}\left\{\underset{j\in {\cal D}}{\bigcup} {\cal E}^{(j)}_1\right\} \\\nonumber &\leq d\text{ Pr} \left\{ |{\cal S}_j(\underline{Y_1})| \leq |{\cal T}_j(\underline{Y_1})|[1-q_R(1+\tau))]|j \in {\cal D}\right\}\\
\nonumber
&= d \sum_{t=0}^{T_1} {T_1 \choose t} p_1^t (1-p_1)^{T_1-t} \sum_{v=tq_R(1+\tau)}^{t} {t \choose v} q^v (1-q)^{t-v}\\ 
\nonumber
&= d \sum_{t=0}^{T_1} {T_1 \choose t} p_1^t (1-p_1)^{T_1-t} \hspace{-0.4cm} \sum_{v=qt+t(q_R-q+q_R\tau)}^{t} {t \choose v} q^v (1-q)^{t-v}\\
\nonumber
&\leq n^{-\delta}\\
\label{eqn-beta1_UB_0}
&\Leftarrow \beta_1 \geq \frac{\left(\frac{\ln D}{\ln n}+\delta\right)\ln 2}{p_1(1-e^{-2})(q_R-q+q_R\tau)^2}\\
\label{eqn-beta1_UB}
&\Leftarrow \beta_1 \geq \frac{\left(\frac{\ln D}{\ln n}+\delta\right)\ln 2}{p_1(1-e^{-2})(q_R\tau)^2},
\end{align}}where (\ref{eqn-beta1_UB_0}) is obtained following similar steps used in obtaining (\ref{eqn-beta1}), and $q\leq q_R$ ensures that Chernoff-Hoeffding bound is applicable as well as ensures sufficiency of (\ref{eqn-beta1_UB}) to guarantee that (\ref{eqn-beta1_UB_0}) holds for all $r\in[0,R]$. An appropriate choice of $\tau>0$\footnote{As shall be seen later, the chosen value of $\tau$ will also satisfy $q_R(1+\tau)\leq 1$ so that Chernoff-Hoeffding bound gives a non-trivial upper bound. Otherwise, the probability of the event under consideration will be equal to zero.} shall be specified later.

Similarly, to ensure that $\text{Pr}\left\{\underset{j\in {\cal N}}{\bigcup} {\cal E}^{(j)}_2\right\} \leq n^{-\delta}$, from $(\ref{eqn-beta2})$, we have
\begin{align}\label{eqn_beta1_UB2_0}
&\beta_1 \geq \frac{\left(\frac{\ln n-d-r}{\ln n}+\delta\right)\ln 2}{p_1(1-e^{-2})\left(1-q_R(1+\tau)-a\right)^2},
\end{align} where $a$ is defined in (\ref{eqn-define-a}) with $p=p_1$. The term $1-q_R(1+\tau)-a$ is lower bounded as
\begin{align}
\nonumber
&1-q_R(1+\tau)-a =1-q_R-a-q_R\tau \\
\nonumber
&\underset{(a)}{\geq} 1-(1-(1-p_1)^R)-(1-(1-p_1)^D)-q_R\tau\\
\nonumber
&\geq 1-(R+D)p_1-q_R\tau,
\end{align}where the lower bound $(a)$ follows from the fact that $a \leq 1-(1-p_1)^D$. So, to guarantee (\ref{eqn_beta1_UB2_0}) it is sufficient that 
\begin{align} \label{eqn-eqn_beta1_UB2}
&\beta_1 \geq \frac{\left(1+\delta\right)\ln 2}{p_1(1-e^{-2})\left(1-(R+D)p_1-q_R\tau\right)^2},
\end{align}Optimizing the denominators of the above inequality and (\ref{eqn-beta1_UB}) with respect to $\tau$ (by equating the denominators), we have 
\begin{align}\label{eqn-tau2}
 \tau = \frac{1-(R+D)p_1}{2q_R}.
\end{align}Note that this value of $\tau$ is independent of $r$ and $d$ so that the decoding algorithm is also independent of $r$ and $d$. To ensure that $\tau>0$, we must have $0<p_1 <\frac{1}{R+D}$. For this range of $p_1$, optimizing
\begin{align*}
 \frac{1}{p_1(q_R\tau)^2}=\frac{1}{p_1\left(1-(R+D)p_1-q_R\tau\right)^2}
\end{align*} w.r.t. $p_1$ with $\tau$ chosen as in (\ref{eqn-tau2}) yields
\begin{align*}
p_1=\frac{1}{3(R+D)}.
\end{align*}Substituting the chosen values of $p_1$ and $\tau$ in (\ref{eqn-beta1_UB}) and (\ref{eqn-eqn_beta1_UB2}), we have

{\small\begin{align} \label{eqn-beta1_UB_final}
\beta_1 \geq \max \left\{\frac{27(R+D)\left(\frac{\ln D}{\ln n}+\delta\right)\ln 2}{(1-e^{-2})},\frac{27(R+D)\left(1+\delta\right)\ln 2}{(1-e^{-2})}\right\}.
\end{align}}

To satisfy $\text{Pr}\left\{\underset{j\in {\cal N}}{\bigcup} {\cal E}^{(j)}_3\right\} \leq n^{-\delta}$ in the second stage of the decoding algorithm with $T_2=\beta_2 \log n$, we have (\ref{eqn-beta3}), with $p=p_2$. The inequality (\ref{eqn-beta3}) is satisfied for all $(r,d)$ if it is satisfied with the denominator of the RHS evaluated at $r=R,d=1$ and the numerator of the RHS evaluated at $r=d=0$, i.e., 
\begin{align}\label{eqn-beta2_UB_0}
 \beta_2 \geq \frac{(1+\delta)\ln 2}{p_2^2(1-p_2)^R} \geq \frac{(1+\delta)\ln 2}{p_2^2(1-Rp_2)}.
\end{align}Optimizing the above expression for $p_2<\frac{1}{R}$ yields $p_2= \frac{2}{3R}$. Therefore, from the above bound along with the bound for $\beta_2$ obtained from (\ref{eqn-beta4}) to satisfy $\text{Pr}\left\{\underset{j\in {\cal I}}{\bigcup} {\cal E}^{(j)}_4\right\}\leq n^{-\delta}$, we have
\begin{align} \label{eqn-beta2_UB}
 \beta_2 \geq \max \left\{\frac{27}{4}R^2(1+\delta)\ln 2, \frac{3R}{2}\left(\frac{\ln R}{\ln n}+\delta\right)\ln 2\right\}.
\end{align}

From (\ref{eqn-beta1_UB_final}) and (\ref{eqn-beta2_UB}), the scaling of the total number of tests $T=T_1+T_2$ for various regimes of $R$ and $D$ is evaluated as below.
\begin{enumerate}
 \item For $R^2=O(D)$, the number of tests scales as $T=O(D \log n)$.
 \item For $D=o(R^2)$, the number of tests scales as $T=O(R^2 \log n)$.
\end{enumerate}
 
\begin{remark}
For DCP under this set-up, the number of tests required is given by $T=(R+D)\log n$. Similar to the previous set-up (as noted in Remark $1$), here too only the error events ${\cal E}^{(j)}_{1}$ and ${\cal E}^{(j)}_{2}$ matter for DCP and hence, only the bound in (\ref{eqn-beta1_UB_final}) needs to be satisfied.
\end{remark}

\begin{remark}
 We note that if $M$ is chosen to be a single i.i.d. test matrix then, $p_2=p_1$. From (\ref{eqn-beta2_UB_0}), such a choice of $p_2$ would result in $O\left((R+D)^2 \log n\right)$ scaling in the number of tests which is clearly sub-optimal.
\end{remark}

In the next sub-section, we show that the required scaling in the number of tests is indeed close to the lower bound.

\subsection{Lower Bound}
Clearly, the number of tests required for the proposed pooling design is order optimal for both SCP and DCP in the $R^2=O(D)$ regime. The following proposition shows that the number of tests required for SCP in the proposed pooling design in the $D=o(R^2)$ regime exceeds the lower bound for non-adaptive pooling designs by at most a $\log r$ multiplicative factor. 

\begin{proposition}
An asymptotic lower bound on the number of tests for non-adaptive pooling designs for SCP in GTI with knowledge of only upper bounds on the number of inhibitors $R$ and defectives $D$ is given by $\max\left\{\Omega\left((R+D) \log n\right),\Omega\left(\frac{R^2}{\log R} \log n\right)\right\}$.
\end{proposition}
\begin{proof}
Since vanishing error probability for the worst $(r,d)$  combination needs to be ensured as given in (\ref{eqn-ob_UB}), the first lower bound of $\Omega\left((R+D) \log n\right)$ is trivial. The second lower bound is obtained by evaluating the lower bound in Theorem \ref{thm_LB} at $r=R$ and $d=1$ which minimizes the probability of a positive outcome.
\end{proof}

The following result, similar to Corollary \ref{cor1} in the previous set-up, gives a lower bound for DCP under the current GTI set-up.
\begin{corollary}
An asymptotic lower bound on the number of tests for non-adaptive pooling designs for DCP in GTI with knowledge of only upper bounds on the number of inhibitors $R$ and defectives $D$ is given by $\max\left\{\Omega\left(D \log n\right),\Omega\left(\frac{R}{\log R} \log n\right)\right\}$.
\end{corollary}
\begin{proof}
 The first lower bound is trivial and the second one follows by evaluating the lower bound in Corollary \ref{cor1} at $r=R$ and $d=1$.
\end{proof}

Thus, for DCP in GTI with knowledge of only upper bounds on the number of inhibitors and defectives, the proposed pooling design is order optimal in the number of tests in the $R=O(D)$ regime and exceeds the lower bound by a $\log R$ multiplicative factor in the $D=o(R)$ regime.

\section{Conclusion}
Probabilistic non-adaptive pooling design was proposed for SCP (and, as a by-product, for DCP) in the GTI model and a column-matching like decoding algorithm was proposed on the lines of \cite{CJSA} for the following cases.
\begin{itemize}
 \item Exact number of inhibitors and defectives is known.
 \item Upper bounds on the number of inhibitors and defectives are known.
\end{itemize}
In the small inhibitor regime, the proposed pooling design is shown to be order optimal in the number of tests. In the large inhibitor regime, the number of tests required in the proposed pooling design is observed to exceed the lower bound by logarithmic multiplicative factors. Similar gaps between the upper and lower bounds on the number of tests exist even in zero-error SCP and DCP in the GTI model for both adaptive and non-adaptive pooling designs, as observed from Table \ref{tab1}. Also, as seen from Table \ref{tab1}, for $\epsilon$-error the number of tests required is much less compared to the number of tests required for zero-error pooling designs.

As noted in Section \ref{sec4}, inhibitors could be regarded as an asymmetric noise. Hence, a noisy channel coding approach of \cite{AAS} could tighten the lower bound in the large inhibitor regime.
Extensions to other known GTI models like the $k$-inhibitor model and the threshold GTI model with $\epsilon$-error targets are directions worth exploring.

\end{document}